\documentclass[oribibl]{llncs}
\usepackage{epic}
\usepackage{eepic}
\usepackage{epsf}

\newcommand{\real}{\mbox{I}\!\mbox{R}}

\newcommand{\svec}[3]{\left( \begin{array}{c} #1 \\ #2 \\ #3 \end{array} \right)}

\newcommand{\postscript}[2]
{\setlength{\epsfxsize}{#2\hsize}
\centerline{\epsfbox{#1}}}

\spnewtheorem{algorithm}{Algorithm}{\bfseries}{\itshape}

\begin{document}

\bibliographystyle{plain}

\title{Optimal Binary Search Trees with \\ 
       Near Minimal Height}
\author{Peter Becker}
\institute{Faculty of Computer Science \\
   University of Applied Sciences Bonn-Rhein-Sieg\\
   Grantham-Allee 20, 53757 Sankt Augustin, Germany\\
   \email{peter.becker@h-brs.de}
   }

\maketitle

\begin{abstract}
Suppose we have $n$ keys, $n$ access probabilities for the keys, and $n+1$
access probabilities for the gaps between the keys.
Let $h_{\min}(n)$ be the minimal height of a binary search tree for $n$ keys.
We consider the problem to construct an optimal binary search tree with near minimal
height, i.e.\ with height 
$h \leq h_{\min}(n) + \Delta$ for some fixed $\Delta$.
It is shown, that for any fixed $\Delta$ optimal binary search trees
with near minimal height
can be constructed in time $O(n^2)$. This is as fast as in the unrestricted case.

So far, the best known algorithms for the construction of height-restricted
optimal binary search trees have running time $O(L n^2)$,
whereby $L$ is the maximal permitted height.
Compared to these algorithms our algorithm is at least faster
by a factor of $\log_2 n$, because $L$ is lower bounded by $\log_2 n$.
\end{abstract}

\section{Introduction}

Suppose we have $n$ keys, $n$ access probabilities for the keys, and $n+1$
access probabilities for the gaps between the keys.
The problem to construct a binary search tree for these $n$ keys
that minimizes
the expected access time is known as the
\emph{optimal binary search tree problem}.
Knuth presented in \cite{knu71} a well-known dynamic programming algorithm
that solves this problem in $O(n^2)$ time.

Apart from the original problem,
the construction of optimal binary search trees whose heights are restricted 
has been considered in the literature.
By the height restriction the maximum number of comparisons during a search can
be bounded. Thus, an optimal height restricted binary search tree
performs well in both the worst and the average case.
Itai \cite{itai76} and Wessner \cite{wess76} independently discovered 
construction algorithms for height restricted binary search trees.
Their algorithms have running time $O(L n^2)$, where $L$ is the maximal permitted
height.

Let $h_{\min}(n) = \lceil \log_2 (n+1)\rceil$
be the \emph{minimal height of a binary search tree for $n$ keys}.
In this paper, we show that for any fixed $\Delta$ an optimal
binary search tree with height $h \leq h_{\min}(n) + \Delta$
can be constructed in time $O(n^2)$.
This improves the results from Itai and Wessner \cite{itai76,wess76}.
Because $L \geq \lceil \log_2 (n+1) \rceil$, the algorithms of Itai and Wessner
have running time $O(n^2 \log n)$ if we use them to construct optimal search
trees with height $h \leq h_{\min}(n) + \Delta$.

Gagie \cite{gag03, gag05} presents a $O(n)$ time algorithm
for the restructuring of optimal binary search trees.
His algorithm restructures an existing optimal binary search
in such a way that the resulting tree
has nearly optimal height and cost.
In contrast to Gagie's algorithm our algorithm always selects
the best binary search tree from the set of all trees with restricted height.

Other interesting facts about optimal binary search trees can be found in the
article of Nagaraj \cite{nag97}. This article gives a
comprehensive survey about optimal binary search trees.

All algorithms for the construction of optimal binary search trees, whether height
restricted or not, are based on dynamic programming.
They all use step by step construction of larger trees from smaller subtrees.
Instead of step by step construction from smaller subtrees we use a
decision model where the keys are placed by a sequential decision process in such a way
into the tree, that the costs become optimal. This approach is adopted from
the construction algorithm for optimal B-trees \cite{bec94}.

The rest of the paper is structured in the following way: in
Section 2 a formal description of the problem is given.
In Section 3 we present our approach:
the decision model is explained and
the attached dynamic program is formulated.
Section 4 states the solution algorithm
and gives the complexity results.
Section 5 summarizes the results.

\section{The Problem}

Now we give the problem formulation.
We have $n$ keys $k_1 < k_2 \ldots < k_n$
and $2n+1$ probabilities
$\alpha_0 , \beta_1 , \alpha_1 , \beta_2 , \ldots , \beta_n , \alpha_n $.

$\beta_i$ are the \emph{key weights} and $\alpha_j$ are the
\emph{gap weights}.
$\beta_i$ is the probability that key $k_i$ is requested, and
$\alpha_j$ is the
probability, that a search is made for a key $d$ with
$k_j < d < k_{j+1}$.
We assume that we have artificial keys $k_0 = - \infty$ and $k_{n+1} = \infty$.

Let $b_i$ be the level resp.\ the depth of the $i$-th internal node where key
$k_i$ is stored, and
let $a_j$ be the level of the external node for the gap between $k_j$ and $k_{j+1}$.
The root is on level $0$. 
For a binary search tree $T$ we define the
\emph{weighted path length} $\textnormal{wpl}(T)$ by
\[
\textnormal{wpl}(T) := \sum_{i=1}^n \beta_i (b_i+1) + \sum_{j=0}^n \alpha_j a_j
\]
The weighted path length is the expected number
of node visits resp.\ comparisons in a search.

The height $h(T)$ of a tree $T$ is defined as the level of the deepest external node.
The \emph{minimal height} $h_{\min}(n)$ of a binary search tree for $n$ keys
is then given by
\[
   h_{\min}(n) = \lceil \log_2(n+1) \rceil
\]

We want to construct search trees whose heights are nearly minimal.
Let $\Delta \geq 0$ be some fixed value.
The problem is to find a binary search tree $T$
that minimizes the weighted path length
$\textnormal{wpl}(T)$ subject to the constraint
$h(T) \leq h_{\min}(n) + \Delta$. 
Such a tree is denoted as an
\emph{optimal binary search tree with near minimal height}.

\section{Dynamic Programming Model}

We model the process of constructing an optimal binary search tree
with near minimal height
as a decision problem with $n$ stages. For every key $k_i$
we have to decide, on which level this key should be placed.
Whether placing on some level is feasible,
depends on the former decisions for the keys $k_1$ to $k_{i-1}$,
which define a certain state in the decision process.
Then placing the key $k_i$ on any level results in an increasing
weighted path length and a new state. The amount of increasing
as well as the new state depend on our decision.

Using this approach, the
optimal tree is the result of a sequence of optimal
decisions starting in a unique initial state. This leads to
a dynamic program $DP$ of the form $DP = (S_\nu ,A_\nu ,D_\nu ,
T_\nu , c_\nu , C_{n+1})$, where $n$ is the number of the
stages of $DP$, $S_\nu$ is the \emph{state set} of stage
$\nu , 1 \leq \nu \leq n+1$, and $A_\nu$ is the \emph{decision set}
of stage $\nu , 1 \leq \nu \leq n$. The sets
$D_\nu \subseteq S_\nu \times A_\nu$
define the \emph{feasible decisions} for the states of stage $\nu$.
It holds: $(s,a) \in D_\nu$, if and only if $a$
is feasible in state $s$ on stage $\nu$.
The set $D_\nu (s) := \{ a \in A_\nu | (s,a) \in D_\nu \}$
contains all feasible decisions for state $s$ on stage $\nu$.
$T_\nu : D_\nu \rightarrow S_{\nu +1}$ is
the \emph{transition function}.
Making decision $a$ in state $s$ at stage $\nu$ results in state
$T_\nu(s,a)$ at stage $\nu + 1$.
$c_\nu : D_\nu \rightarrow \real$ is the \emph{cost function} of
stage $\nu$. $c_\nu(s,a)$ gives the costs that arise if we
decide to make decision $a$ in state $s$ on stage $\nu$.
$C_{n+1} : S_{n+1} \rightarrow \real$ is the
\emph{terminal cost} function.
$C_{n+1}(s)$ gives the costs that arise if our final state
is $s$.

\begin{figure}[htb]
\postscript{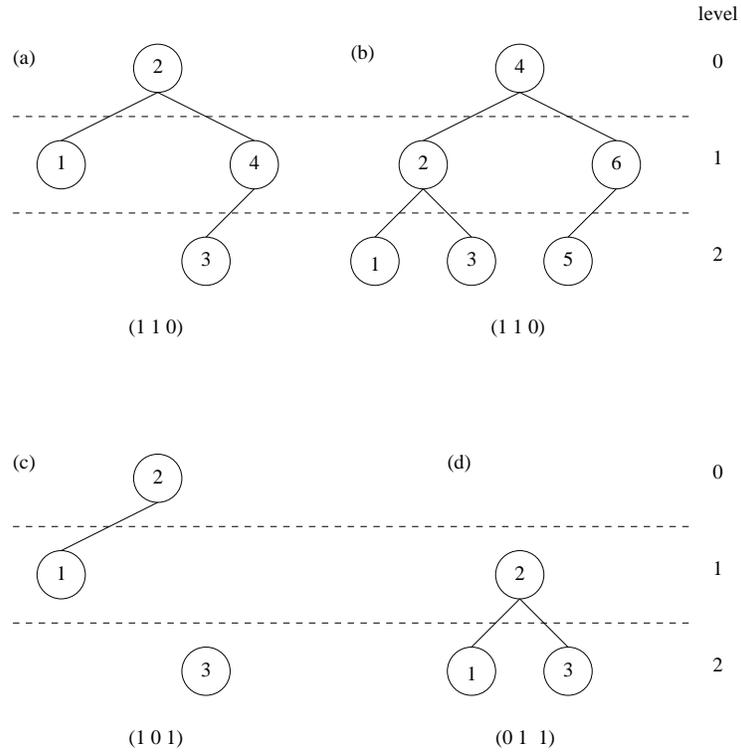}{0.8}
\caption{Tree states in the construction process}
\label{figstates}
\end{figure}

Now we have to define the components of the dynamic program
in such a way that the decision process models the
construction of a binary search tree with restricted height.
First we give the definition of the states.
For motivation take a look at Figure~\ref{figstates}.
Suppose we have $h_{\max} := h_{\min}(n) + \Delta = 3$, 
that means we can place the keys on levels from $0$ to $2$.

For a correct placing
of a key in the partial tree only the
rightmost path fragments from the actual root to the node that contains the largest key
is relevant. Due
to this fact we can represent a state $s \in S_\nu$
by a binary vector with $h_{\max}$ components.
We number the vector components from $0$ to $h_{\max} -1$.
Vector component $s_i$ is related to level $i$.
\[
    s = \svec{s_0}{\vdots}{s_{h_{\max}-1}} \textnormal{ with } s_i \in \{0,1\}.
\]
Each vector component $s_i$ determines, whether the level $i$ 
in the rightmost path is occupied.
More formally, vector component $s_i$ is $1$ if and only if the largest
key on level $i$ is greater than any key on the levels from $0$ to $i-1$.
For instance the state $s$ resulting from tree (a)
in Figure~\ref{figstates} is represented
by
\[
   s = \svec{1}{1}{0}
\]
and the state $s'$ resulting from tree (c) by 
\[
   s' = \svec{1}{0}{1}
\]
Observe, that different trees may have the same associated states.
For instance the trees (a) and (b) of Figure \ref{figstates} are both
represented by the same state.

The set $S_\nu$ is defined to be the set of all vectors
that are possible after the assignment of $\nu - 1$ keys.
The initial state set $S_1$ consists of a single state:
\[
   S_1 := \left\{ \svec{0}{\vdots}{0} \right\}
\]

A decision is characterized by the level on which a key is placed.
So we define $A = A_\nu = \{0 ,\ldots, h_{\max}-1 \}$.
Making decision $a$ means that the corresponding key is placed on level $a$.
For instance, the tree (a) in Figure~\ref{figstates} is constructed by the
decision sequence $DS = (1,0,2,1)$.

\begin{figure}[hbt]
\postscript{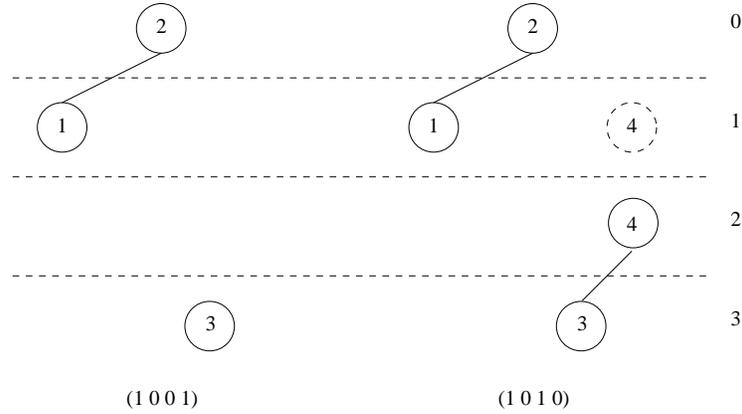}{0.8}
\caption{Feasible and infeasible decision}
\label{infeasible}
\end{figure}

Let $s = (s_0 ,\ldots, s_{h_{\max}-1})$ be a state.
A feasible decision $a$ for state $s$
has to fulfill the following conditions:
\begin{itemize}
\item[(i)]
   We can place keys only on unoccupied levels:
   \[
       s_a = 0
   \]
\item[(ii)]
   If a key is placed above some path fragment, this path fragment
   has to be the deepest path
   fragment and the key has to be placed directly above this path fragment:
   \[
      \not\exists i,j : a < i < j \textnormal{ and } s_i = 0 \textnormal{ and } s_j = 1
   \]
\end{itemize}
Condition (i) is obvious. Figure~\ref{infeasible} demonstrates condition (ii).
The next key $k_4$ has to be placed on level $2$, because $k_3$ becomes the left son
of $k_4$. If we place $k_4$ on level $1$, the left son would not be on the next
deeper level.

So we can define
\[
   D_\nu := \{ (s,a) | s \in S_\nu, a\;\textnormal{fulfills (i) and (ii)} \}
\]
Observe that the feasible decisions of a state $s$ are independent
of the stage $\nu$. So we define
\[
    D(s) := \{ a \in A | a\;\mbox{fulfills (i) to (ii)} \}
\]
as the \emph{set of feasible decisions for state $s$}.
For every binary search tree (with near minimal height) there exists a unique
feasible decision sequence that constructs the tree. As an
example see the decision sequence to construct tree (a) of Figure~1 (see
above). Using this definition each feasible decision sequence
leads to trees that are valid binary search trees with the exception of
the rightmost path.
Trees with invalid rightmost path on stage $n+1$
are filtered by the terminal cost function $C_{n+1}$ (see below).

\begin{figure}[htb]
\postscript{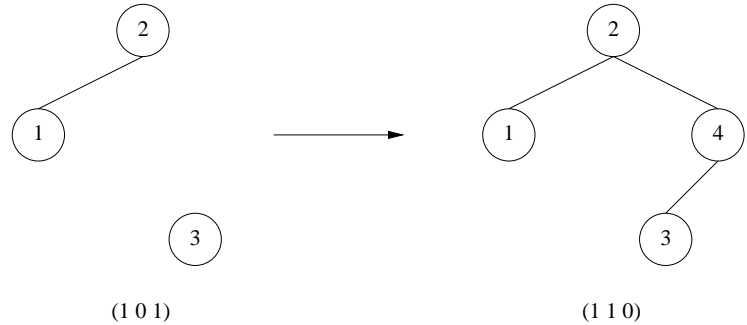}{0.8}
\caption{Example for a transition}
\label{transition}
\end{figure}

Making a decision $a$ has two effects. First, the
level $a$ of the rightmost path becomes occupied and
second, the levels from $a+1$ to $h_{\max}-1$ become unoccupied.
So the definition for the transition function is:
\[
   T(s,a) := T_\nu(s,a) = \left(\begin{array}{c} 
      s_0 \\ 
      \vdots \\
      s_{a-1} \\
      1 \\
      0 \\
      \vdots\\
      0 \end{array}\right)
\]
Figure~\ref{transition} shows an example for a single transition. 
The following state and decision sequence shows the transitions from the
initial state to the right tree of Figure~\ref{transition}.
\[
\svec{0}{0}{0} \stackrel{a=1}{\longrightarrow}
\svec{0}{1}{0} \stackrel{a=0}{\longrightarrow}
\svec{1}{0}{0} \stackrel{a=2}{\longrightarrow}
\svec{1}{0}{1} \stackrel{a=1}{\longrightarrow}
\svec{1}{1}{0}
\]

If we have a state $s \in S_\nu$, we can deduce from $s$ the preceding decision, 
i.e.\ the decision on stage $\nu -1$ that induced $s$.
Take a look at the transition function $T(s,a)$:
the largest $i$ with $s_i = 1$ defines this preceding decision.
\[
   \textnormal{precdec}(s) := \left\{
   \begin{array}{ll}
   0 & \textnormal{ if } s_0 = \cdots = s_{h_{\max}-1} = 0 \\
   \max \{ 0\leq i \leq h_{\max}-1 | s_i = 1 \} & \textnormal{ otherwise}
   \end{array}
   \right.
\]

Our cost function $c_\nu(s,a)$
has to consider two aspects: the level of key $k_\nu$ and
the level of the gap $(k_{\nu-1},k_\nu)$. The first is simple: the
level of key $k_\nu$ is determined by the decision $a$.
With the following Lemma, we are able the determine the level of the gap $(k_{\nu-1},k_\nu)$.
\begin{lemma}
\label{gaplevel}
Let $\textnormal{klevel}(k_\nu)$ denote the level of key $k_\nu$ and
let $\textnormal{glevel}(k_{\nu-1},k_\nu)$
denote the level of the gap $(k_{\nu-1},k_\nu)$. Then we have
\[
   \textnormal{glevel}(k_{\nu-1},k_\nu)= 1 + \max\{\textnormal{klevel}(k_{\nu-1}),
                                                   \textnormal{klevel}(k_\nu)\}
\]
\end{lemma}

\begin{proof}
Adjacent keys cannot be on the same level. So we have either
$\textnormal{klevel}(k_{\nu-1}) < \textnormal{klevel}(k_\nu)$
or
$\textnormal{klevel}(k_{\nu-1}) > \textnormal{klevel}(k_\nu)$.

In the case of $\textnormal{klevel}(k_{\nu-1}) < \textnormal{klevel}(k_\nu)$,
the key $k_\nu$ is in the right subtree of key $k_{\nu-1}$ and the
gap $(k_{\nu-1},k_\nu)$ is the left son of the node that contains $k_\nu$.
In the other case
the key $k_{\nu-1}$ is in the left subtree of key $k_\nu$ and the
gap $(k_{\nu-1},k_\nu)$ is the right son of the node that contains $k_{\nu-1}$.
In both cases the equation of Lemma~\ref{gaplevel} is valid.
\end{proof}

The cost functions $c_\nu(s,a)$ are defined by:
\[
   c_\nu(s,a) := 
                (1+\max\{\textnormal{precdec}(s),a\})\cdot\alpha_{\nu -1} + 
                (a+1)\cdot\beta_\nu
\]
This definition utilizes Lemma~\ref{gaplevel}: $\textnormal{klevel}(k_{\nu-1})$
is equivalent to $\textnormal{precdec}(s)$ and $\textnormal{klevel}(k_\nu)$ to
the decision $a$.

The terminal costs $C_{n+1}$ model whether our final
state fulfills the tree conditions.
In particular, we have to check whether the right most path contains
unoccupied levels above occupied levels.
For instance, tree (c) of Figure~\ref{figstates} is not a valid search tree
because level~$1$ is not occupied but level~$2$ is.
We have:
\[
   C_{n+1}(s) = \left\{
   \begin{array}{ll}
      (1 + \textnormal{precdec}(s))\cdot\alpha_n & \textnormal{ if } 
         s_0 = 1 \textnormal{ and } 
	 \not\exists i < j: s_i = 0 \wedge s_j = 1 \\
      \infty & \textnormal{ otherwise}
   \end{array}
   \right.
\]
To check whether there exists an unoccupied level we use an adaption of
condition (ii) of the feasible decision set $D(s)$. If the root level
is occupied and there exists no unoccupied level
above an occupied level
the terminal costs consist of the access probability $\alpha_n$
of the last gap  multiplied
by the level of key $k_n$ plus $1$.

Now the definition of the dynamic program $DP$ is
complete. Using this definition the optimization problem is
\[
F :=
\sum_{\nu = 1}^n c_\nu (s_\nu , a_\nu ) + C_{n+1} (s_{n+1})
\rightarrow \min
\]
subject to:
\[
\begin{array}{ll}
&   s_1 = (0 \cdots 0) \\
&   a_\nu \in D (s_\nu ), 1 \leq \nu \leq n \\
&   s_{\nu + 1} = T(s_\nu , a_\nu ), 1 \leq \nu \leq n
\end{array}
\]
The value $F$ of the objective function yields the
minimum weighted path length and the tree is given by
the optimal sequence $(a_1,\ldots,a_n)$ of feasible decisions.

\section{Algorithm and Complexity}

For the solution of this optimization problem we use a common
dynamic programming algorithm, cf.\ \cite{nm02}.

\begin{algorithm}
\label{dpalgo}\mbox{}\\
\begin{tabular}{rl}
(0) & /* Initialization */ \\
(1) & {\bf forall} $s \in S_{n+1}$ \\
(2) & \hspace*{0.5cm} $V_\nu(s) \leftarrow C_{n+1}(s)$ \\
(3) & /* Backward Computation */ \\
(4) & {\bf for} $\nu \leftarrow n$ {\bf downto} $1$ {\bf do} \\
(5) & \hspace*{0.5cm} {\bf forall} $s \in S_\nu$ {\bf do} \\
(6) & \hspace*{1.0cm} $V_\nu(s) \leftarrow \infty$ \\
(7) & \hspace*{1.0cm} $\pi_\nu(s) \leftarrow \textnormal{undefined}$ \\
(8) & \hspace*{1.0cm} {\bf forall} $a \in D(s)$ {\bf do} \\
(9) & \hspace*{1.5cm} {\bf if} $c_\nu(s,a) + V_{\nu+1}(T(s,a)) < V_\nu(s)$ {\bf then} \\
(10) & \hspace*{2.0cm} $V_\nu(s) \leftarrow c_\nu(s,a) + V_{\nu+1}(T(s,a))$ \\
(11) & \hspace*{2.0cm} $\pi_\nu(s) \leftarrow a$ \\
(12) & /* Forward Computation */ \\
(13) & $s \leftarrow (0 \cdots 0) $ \\
(14) & $F \leftarrow  V_1(s)$ \\
(15) & {\bf for} $\nu \leftarrow 1$ {\bf to} $n$ {\bf do} \\
(16) & \hspace*{0.5cm} $a_\nu \leftarrow \pi_\nu(s)$ \\
(17) & \hspace*{0.5cm} $s \leftarrow T(s, a_\nu)$ \\
\end{tabular}
\end{algorithm}

$V_\nu(s)$ is the \emph{value function} which represents the minimal costs to reach a terminal state
from state $s$ on stage $\nu$. In line (1) and (2) we initialize the value function with
the terminal costs.
$\pi_\nu(s)$ represents the optimal decision for state $s$ on stage $\nu$. 
The value function $V_\nu(s)$ and the optimal decision
$\pi_\nu(s)$ is
determined by the \emph{Bellman equation}
\[
    V_\nu(s) = \min_{a \in D(s)} \{c_\nu(s,a) + V_{\nu+1}(T(s,a)) \}
\]
which is solved for all states on
all stages in lines (4) to (11).

After the backward computation terminates, the $\pi_\nu$ define an
\emph{optimal policy}. To get the
optimal decision sequence
we apply the $\pi_\nu$ in a forward computation (line (13) to (17))
beginning with our initial state. As a result
the $a_\nu$ represent the decision sequence to build an optimal tree and the value of $F$ is the weighted
path length of the optimal tree.

With the decision sequence
$DS = (a_1 ,\ldots, a_n)$ that defines the optimal binary search tree
we are able to build the corresponding tree in linear time,
as for each key $k_\nu$ the level where $k_\nu$ has to be placed
is given by the decision $a_\nu$.

\begin{example}
Suppose we have keys $k_1,\ldots,k_4$ with access probabilities $\beta_1 = \frac{3}{16},
\beta_2 = \frac{1}{16}, \beta_3 = \frac{1}{2}, \beta_4 = \frac{1}{4}$ and
$\alpha_0 = \cdots = \alpha_4 = 0$. Let $\Delta=0$, that means we have
to construct a tree of height $\lceil \log_2(5)\rceil = 3$.

\begin{figure}[htb]
\begin{center}
\include{exproblem}
\end{center}
\caption{State space for the example problem}
\label{ssexample}
\end{figure}

\begin{figure}[htb]
\postscript{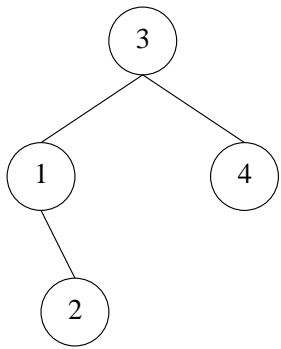}{0.3}
\caption{Optimal binary search tree for the example problem}
\label{ssexampletree}
\end{figure}

Figure~\ref{ssexample} shows the search graph for this problem.
The number adjacent to an arc represents the cost $c_\nu(s,a)$ of the corresponding transition.
The terminal costs $C_5(s)$ are shown below the states of state set $S_5$ and the value function $V_\nu(s)$
is shown
right beside the states for the state sets $S_1$ to $S_4$. Observe, that the the value function
of state $(1,1 1) \in S_4$ yields $\infty$ because of an empty decision set.

The best decision sequence $DS=(1,2,0,1)$ is given by the bold arcs. Its overall cost is
$\frac{25}{16}$, that means the corresponding
optimal binary search tree has a weighted path
length of $\frac{25}{16}$. Figure~\ref{ssexampletree}
shows the corresponding tree.
\end{example}


Our complexity results are based on bounds for the cardinality of the state sets $S_\nu$ and the
decision sets $D_\nu$.

\begin{theorem}
For all state sets $S_\nu\,(\nu=1,\ldots,n+1)$ we have:
\[
|S_\nu | \leq 2^{\Delta+1} (n+1)
\]
\end{theorem}

\begin{proof}
Let $h_{\max}(n) := h_{\min}(n)+\Delta$ and
$S:=\{0,1\}^{h_{\max}(n)}$. With these definitions we get
\[
   |S_\nu | \leq |S| = 2^{h_{\max}(n)} = 2^{h_{\min}(n)+\Delta}
\]
Using $h_{\min}(n) = \lceil \log_2(n+1) \rceil$ we get
\begin{eqnarray*}
   |S| & \leq & 2^{\lceil \log_2(n+1) \rceil + \Delta} \\
       & \leq    & 2^{\Delta+1}\cdot 2^{\log_2(n+1)} \\
       & =       & 2^{\Delta+1}\cdot(n+1)
\end{eqnarray*}
\end{proof}

\begin{corollary}
For any fixed $\Delta$ the cardinality of the state sets $S_\nu$
is bounded by $O(n)$.
\end{corollary}

\begin{theorem}
For all feasible decision sets $D_\nu\,(\nu=1,\ldots,n+1)$ we have:
\[
   |D_\nu| \leq 2^{\Delta+2} (n+1) 
\]
\end{theorem}

\begin{proof}
Let $h_{\max}(n) := h_{\min}(n)+\Delta$,
$S:=\{0,1\}^{h_{\max}(n)}$ and $D := \{(s,a)| s \in S, a \textnormal{ is feasible for } s\}$.
With these definitions we get $|D_{\nu}| \leq |D|$ for all $\nu=1,\ldots,n$.

How many feasible decisions exists for a state $s \in S$? Take a look at condition (ii) in the
definition of $D(s)$ (see Section~3).
If $s_{h_{\max}-1} = 1$ there is at most one feasible decision $a$, which is
determined by the highest index $a$ with $s_a = 0$. That means, that half of all the states in $S$ have only
one feasible decision. States with $s_{h_{\max}-1} = 0$ and $s_{h_{\max}-2} = 1$, which comprise a quarter
of all states in $S$, have at most two decisions. Generalized, $\frac{1}{2^k}|S|$ states 
of all the states in $S$ have $k$ feasible
decisions. We get:
\begin{eqnarray*}
   |D_{\nu}| & \leq & |D| \\
             & \leq & 1\cdot\frac{1}{2}|S| + 2\cdot\frac{1}{4}|S| +
	              3\cdot\frac{1}{8}|S| + \cdots \\
	     & \leq    & \sum_{k=0}^{\infty} \frac{k}{2^k} \cdot |S| \\
             & =       & \left(\sum_{k=0}^{\infty} \frac{k+1}{2^k} -
	                 \sum_{k=0}^{\infty} \frac{1}{2^k}\right) \cdot |S| \\
	     & =       & \left(\frac{1}{(1-\frac{1}{2})^2} - 
	                       \frac{1}{1-\frac{1}{2}}\right) \cdot |S| \\
	     & =       & 2\cdot |S| \\
	     & \leq    & 2^{\Delta+2} (n+1)
\end{eqnarray*}
\end{proof}

\begin{corollary}
For any fixed $\Delta$ 
Algorithm~\ref{dpalgo} constructs an optimal binary search with
height $h \leq h_{\min}(n)+\Delta$ in time $O(n^2)$.
\end{corollary}

\begin{proof}
We have to iterate over the $n$ stages from $n$ down to $1$.
In doing so, the cardinality of 
each state set $S_\nu$ and each feasible decision set $D_\nu$
is bounded by $O(n)$ for fixed $\Delta$. All operations can be executed in
constant time. It follows, that the overall running time is $O(n^2)$.
\end{proof}

\section{Summary}

We have presented a quadratic time algorithm to compute optimal binary search trees with near
minimal height, i.e.\ with height $h \leq h_{\min}(n) + \Delta$ and fixed $\Delta$.
The algorithm was adopted  
from the construction algorithm for optimal B-tress. The construction process was modeled by a
decision oriented dynamic program: In the model we have to decide key by key, on which level the
key should be placed. The tree conditions are represented by additional constraints and a
terminal cost function.

It seems to be easy to apply this approach to other kinds of trees.
By applying the construction algorithm of \cite{bec94},
it should be possible to construct
optimal B-trees with near minimal height and fixed order in quadratic time, too.
The construction of unrestricted optimal B-trees needs time $O(n^{2+\frac{\log 2}{\log k+1}})$.
A generalization of the binary tree model to multiway trees of a fixed order
should also lead to a quadratic time algorithm in constrast to the cubic time
algorithms for the unrestricted case \cite{got81, vkw80}. This means for both cases, that
optimal trees with near minimal height can be constructed faster than unrestricted trees. 
If we consider that optimal trees have typically a low height, the approach of height restriction may lead to
fast construction algorithms, which generate optimal trees with high probability.

\bibliography{paper}

\end{document}